\documentclass[12pt,reqno]{amsart}
\usepackage{lmodern,amsmath,amssymb,amsthm,mathtools,booktabs,longtable}
\usepackage[margin=1in]{geometry}
\usepackage{colonequals}
\usepackage[colorlinks=true,linkcolor=blue,citecolor=blue,urlcolor=blue]{hyperref}

\numberwithin{equation}{section}
\newtheorem{theorem}{Theorem}[section]
\newtheorem{lemma}[theorem]{Lemma}

\newtheorem{remark}[theorem]{Remark}
\newcommand{\A}{\mathcal A}
\newcommand{\R}{\mathbb R}
\newcommand{\E}{\mathbb E}
\renewcommand{\P}{\mathbb P}

\renewcommand{\subset}{\subseteq}
\allowdisplaybreaks[2]
\title{On the Propeller Conjecture}
\author{Steven Heilman}

\thanks{
Email: stevenmheilman@gmail.com\\
Supported by NSF Grant CCF AF 2448108.  
Department of Mathematics, University of Southern California, Los Angeles, CA 90089}

\date{\today}
\begin{document}
\begin{abstract}
We simplify and improve OpenAI's proof of the propeller conjecture of Khot and Naor from 2008: the sum of the squared lengths of the Gaussian moments of any partition of Euclidean space is at most $9/(8\pi)$.  We show equality only occurs for three 120 degree sectors in the plane (up to symmetries, measure zero changes, and taking orthogonal cross products with Euclidean space).  We also remove computer assistance from the three-dimensional case of the conjecture, which previously used an epsilon net traversal.  
\end{abstract}
\maketitle

\section{Statement and proof structure}
Let $\|\cdot\|$ denote the standard Euclidean norm on $\R^n$.  Let $\gamma_n$ be standard Gaussian probability measure on $\mathbb R^n$, with density $\exp(-\|x\|^2 /2)/(2\pi)^{n/2}$.  Denote $S^{n-1}\colonequals\{x\in\R^{n}\colon\|x\|=1\}$.  For a finite measurable partition $\mathcal A=(A_1,\ldots,A_N)$ of $\R^n$, put for any $1\leq i\leq N$,
\begin{equation}\label{mdef}
\begin{aligned}
z_i\colonequals\int_{A_i}x\,d\gamma_n(x)&\in\R^n,\qquad p_i\colonequals\gamma_n(A_i),\qquad F(\A)\colonequals\sum_{i=1}^N\|z_i\|^2,
\\
C_N&\colonequals\sup_{n\geq1}\sup_{\A}F(\A),
\qquad
C\colonequals\sup_{N\geq1}C_N.
\end{aligned}
\end{equation}

\begin{theorem}[Propeller Conjecture {\cite{KN,khot13}}]\label{intro:main}
$C=\frac{9}{8\pi}$.
\end{theorem}
\begin{remark}\label{uniquerk}
Moreover, if $F(\A)=9/(8\pi)$, then up to measure zero changes, rotations, and relabeling, we have $n\geq2$ and $\A=(B_1\times \R^{n-2},B_2\times \R^{n-2},B_3\times \R^{n-2},\emptyset,\ldots,\emptyset)$, where $B_1,B_2,B_3$ are each 120 degree cones centered at the origin.
\end{remark}

\begin{figure}[h]
\centering
\def\svgwidth{.4\textwidth}
\hspace{2cm}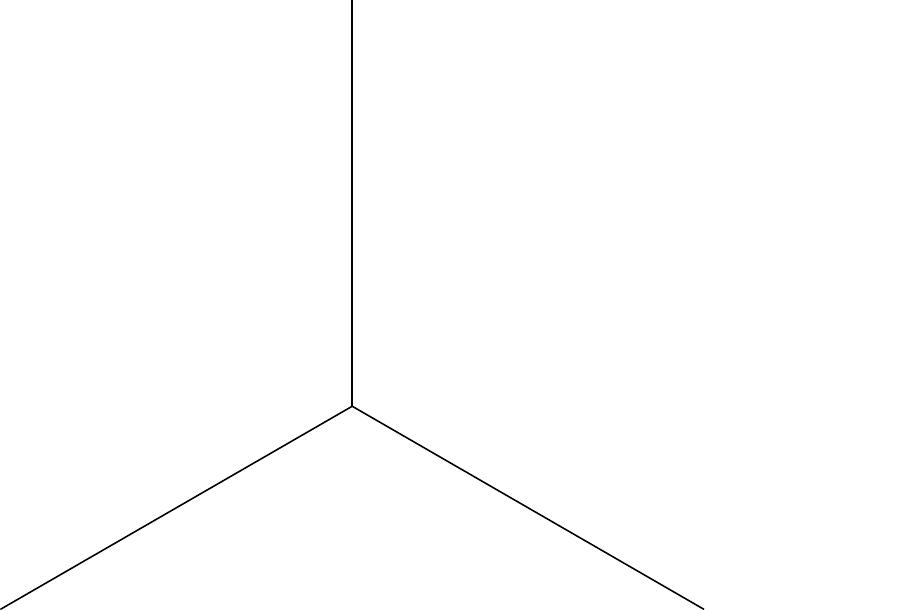
\caption{The propeller partition $B_1,B_2,B_3\subset\R^2$ that maximizes $F(\A)$.
}
\end{figure}

The computation of $C_2=1/\pi$ is an elementary rearrangement argument \cite[Corollary 3.4]{KN}.  $C_3=9/(8\pi)$ was demonstrated in \cite[Corollary 3.4]{KN} via Lagrange multipliers, which in particular implies that
\begin{equation}\label{cst}
C\geq C_N\geq 9/(8\pi),\qquad\forall\,N\geq3.
\end{equation}
$C_4=9/(8\pi)$ was proven with a computer-assisted epsilon net traversal in \cite{HJN}.  In this paper, we prove $C_4=9/(8\pi)$ without computer assistance in Section \ref{sec:analytic-three}.  As in \cite{openai26}, we then show that $C_N=C_4$ for all $N>4$, which proves Theorem \ref{intro:main}.

Part of the interest of Theorem \ref{intro:main} is its different formulations, one of which is a variant of the positive semidefinite Grothendieck inequality \cite{groth53,rietz74}.

\begin{theorem}[{\cite{KN,khot13}}]\label{kntheorem}  The following are equivalent:
\begin{itemize}
\item (Geometric Form) $C=9/(8\pi)$.  That is, for any $n,N\geq1$, for any partition $(A_1,\ldots,A_N)$ of $\R^n$, we have $\sum_{i=1}^{N}\|\int_{A_i}xd\gamma_{n}(x)\|^{2}\leq 9/(8\pi)$, 
\item (Probabilistic Form) Let $(g_1,\ldots,g_n)$ be a mean zero Gaussian random vector.  Then
\begin{equation}\label{vecform}
\quad
(\E \max_{i=1,\ldots,n}g_i)^2\leq\frac{9}{8\pi}\sum_{i=1}^{n}\E g_i^2,\quad\forall\, n\geq1.
\end{equation}
\item (Analytic Form) Let $(a_{ij})$ be an $n\times n$ real symmetric positive semidefinite matrix with $\sum_{i,j=1}^n a_{ij}=0$.  Then, $\forall$ $k\geq3$ and for every $v_1,\ldots,v_k\in S^{k-1}$ with $\sum_{i=1}^{k}v_i=0$,
\begin{equation}\label{svdgro}
\max_{w_1,\ldots,w_n\in S^{n-1}}
\sum_{i,j=1}^n a_{ij}\langle w_i,w_j\rangle
\leq\frac{8\pi}{9}\Big(1-\frac{1}{k}\Big)\cdot
\max_{\sigma\colon\{1,\ldots,n\}\to\{1,\ldots,k\}}
\sum_{i,j=1}^n a_{ij}\langle v_{\sigma(i)},v_{\sigma(j)}\rangle.
\end{equation}
Moreover, the constant $\frac{8\pi}{9}(1-1/k)$ cannot be made smaller.
\item (Complexity Theoretic Form)  Consider the following optimization problem. The input is an $n\times n$ symmetric positive semidefinite centered matrix $(a_{ij})$, and also a $k\times k$ symmetric positive semidefinite matrix $(b_{ij})$ with $k\geq3$, $b_{ii}=1$ $\forall$ $1\leq i\leq k$ and $\sum_{i,j=1}^{k}b_{ij}=0$. The goal is to compute in polynomial time the partition $S_1,\ldots,S_k$ of $\{1,\ldots,n\}$ that maximizes
$\sum_{i,j=1}^k (\sum_{p,q\in (S_i\times S_j)} a_{pq})b_{ij}$. Then the NP-hardness threshold of this problem is $\frac{8\pi}{9}(1-\frac{1}{k})$.  That is, there is an algorithm that can solve this problem within a multiplicative approximation of $\frac{8\pi}{9}(1-\frac{1}{k})$ in polynomial time, and obtaining any smaller approximation factor is NP-hard.
\end{itemize}
\end{theorem}

\begin{remark}
The uniqueness from Remark \ref{uniquerk} implies analogous uniqueness statements for the probabilistic and analytic forms in Theorem \ref{kntheorem}.  For example, equality holds in \eqref{vecform} only when $n\geq3$ and the covariance matrix of $(g_1,\ldots,g_n)$ is (up to reordering) a multiple of
\[
\frac{1}{2}\begin{pmatrix}
2 & -1 & -1 & 0 & \cdots & 0\\
-1 & 2 & -1 & 0 & \cdots & 0\\
-1 & -1 & 2 & 0 & \cdots & 0\\
0 & 0 & 0 & 0 & \cdots & 0\\
\vdots &\vdots  &\vdots  & \ddots & \vdots & 0\\
0 & 0 & 0 & 0 & \cdots & 0\\
\end{pmatrix}
\]
Asymptotic equality in \eqref{svdgro} only occurs as $n\to\infty$ for fixed $v_1,\ldots,v_k\in S^{k-1}$ when $\langle v_i,v_j\rangle=-1/(k-1)$, $\forall$ $1\leq i<j\leq k$, and when $(a_{ij})$ is not the zero matrix.  For justification, see \cite[Sections 3.1, 4]{khot13}, \cite[page 10]{KN} and the proof of Remark \ref{sdrk} in Section \ref{fd:foundations}.
\end{remark}

The geometric form of Theorem \ref{kntheorem} is closely related to the Standard Simplex Conjecture \cite{isaksson11}, which has been verified for a majority of correlation parameters in \cite{heilman26} for three sets, although the four set and higher cases remain open.  The principal difference is that the Standard Simplex Conjecture constrains the measures of the partition elements, whereas the Propeller problem does not.  (Also \cite{isaksson11} concerns noise stability, a generalization of $F(\A)$.)

The probabilistic form of Theorem \ref{kntheorem} is similar to but distinct from the recently resolved Simplex Mean Width Conjecture \cite{mulgund26}, and the Sudakov-Fernique inequality \cite{sudakov71,sudakov76,fernique75}.  In particular, the variance constraints are different, e.g. the former says:
let $n\geq2$ and let $(g_1,\ldots,g_n)$ be a mean zero Gaussian random vector with $\E g_i^2=1$ $\forall$ $1\leq i\leq n$.  Then $\E \max_{i=1,\ldots,n}g_i$ is largest when $\E g_i g_j=-1/(n-1)$ for all $1\leq i<j\leq n$.

The complexity theoretic form of Theorem \ref{kntheorem} originally assumed the Unique Games Conjecture.  Since that conjecture is now proven \cite{openai26}, we do not need it as an assumption.  This optimization problem, known as the kernel clustering problem from machine learning, originated in \cite{song07}.

The analytic form of Theorem \ref{kntheorem} contrasts with the positive semidefinite Grothendieck inequality with sharp constant \cite{rietz74}, which is the analogue of \eqref{svdgro} when $k=2$:

\[
\max_{w_1,\ldots,w_n\in S^{n-1}}
\sum_{i,j=1}^n a_{ij}\langle w_i,w_j\rangle
\leq\frac{\pi}{2}\cdot
\max_{y_{1},\ldots,y_{n}\in\{-1,1\}}
\sum_{i,j=1}^n a_{ij}y_iy_j.
\]

In fact, \cite{khot13} studied a more general problem:
\begin{theorem}[{\cite{khot13}}]\label{knthm}
Let $(a_{ij})$ be an $n\times n$ real symmetric positive semidefinite matrix.  Let $k\geq2$ and let $v_1,\ldots,v_k\in \R^{k}$ be not all equal.  Let $B=(b_{ij})\colonequals(\langle v_i,v_j\rangle)$.  Then
\[
\max_{w_1,\ldots,w_n\in S^{n-1}}
\sum_{i,j=1}^n a_{ij}\langle w_i,w_j\rangle
\leq\frac{1}{C(B)}\cdot
\max_{\sigma\colon\{1,\ldots,n\}\to\{1,\ldots,k\}}
\sum_{i,j=1}^n a_{ij}\langle v_{\sigma(i)},v_{\sigma(j)}\rangle,\,\,\text{where}
\]
\[
C(B)\colonequals\sup_{\A=(A_1,\ldots,A_k)\subset\R^{k-1}}\sum_{i,j=1}^{k}b_{ij}\Big\langle\int_{A_i}xd\gamma_{k-1}(x),\,\int_{A_j}xd\gamma_{k-1}(x)\Big\rangle>0.
\]
Moreover, the constant $1/C(B)$ cannot be improved in the above inequality, even under the additional assumption that $\sum_{i,j=1}^n a_{ij}=0$.
\end{theorem}
Theorem \ref{intro:main} then corresponds to finding the supremum over $k$ of $C(B)$ when $B$ is the $k\times k$ identity matrix.  For general $B$, there is no conjectured simple formula for $C(B)$.
%
%

\begin{remark}\label{sdrk}
\eqref{svdgro} follows from Theorems \ref{knthm} and \ref{intro:main} despite the different assumptions.
\end{remark}
We postpone the proof of Remark \ref{sdrk} to Section \ref{fd:foundations}.

\subsection{Sketch of the Proof of the Main Theorem}

The proof proceeds in two steps.
\begin{itemize}
\item[(i)] If $\A$ maximizes $F$ in $\R^n$, then $\A$ must be a partition with at most four sets of positive measure.  Using known reductions from \cite{KN}, this implies that we may assume $\A\subset\R^3$.  This is proven in Lemma \ref{dimlem}.
\item[(ii)] In $\R^3$ there are only two nondegenerate critical points of $F$ with four sets of positive measure (up to rotations, measure zero changes, and relabeling the sets), and both critical points satisfy $F(\A)<9/(8\pi)$.  This is proven in Section \ref{sec:analytic-three}.  (Here nondegenerate means that the span of the vectors $z_1,\ldots,z_4$ has dimension $3$.)
\end{itemize}
Steps (i) and (ii) prove Theorem \ref{intro:main}, since they together imply that $C_N\leq C_4$ for $N>4$.

The proof of (i) uses a sequence of inequalities that must hold for maximizing critical points.  The argument of \cite{HJN} used a similar strategy, although some parts of \cite{HJN} were specific to $\R^3$ and the strategy was used differently.

The proof of (ii) begins with the reformulated critical point condition \cite[Lemma 3.14, Corollary 3.15]{HJN}, itself inspired by rubber band configurations on the sphere \cite{leech57,con82,linial88,cohn10}.  We then completely classify both nondegenerate critical points for four sets in $\R^3$.  The proof uses some ideas from \cite[Theorem 3.2]{laurent}.

\subsection{Future Directions}  Theorem \ref{intro:main} asserts that $C(I_k)=C(I_3)$ for all $k\geq3$, where $I_k$ is the $k\times k$ identity matrix.  In other words, the computation of $C(I_k)$ reduces to a lower-dimensional optimization problem.  If $B$ is close to the identity matrix, it should be the case that $C(B)$ also reduces to a lower-dimensional optimization problem.  Making this statement precise with a good notion of closeness remains an open problem.

The measure-constrained version of the Propeller Problem still seems out of reach of current methods.  This problem asks to maximize $F(\A)$ in $\R^n$ subject to the constraint that $\gamma_n(A_i)=1/N$ for all $1\leq i\leq N$.  The case $N=n+1$ is most interesting.

The proof of the three-dimensional case of the Propeller Problem demonstrates that the nondegenerate critical points of $F$ are identical to the nondegenerate critical points of four points on the sphere $S^2$ connected by rubber bands \cite[Lemma 3.14, Corollary 3.15]{HJN}.  Lemma \ref{an:classification} gives a complete classification of these nondegenerate critical points.  It would be interesting to try to generalize this result to at least $n+1$ points on the sphere $S^{n-1}$.

\subsection{Statement on AI Use}

GPT 6 was used to prepare this manuscript, which is based upon the solution of the Propeller Conjecture from \cite{openai26}.  Before the release of \cite{openai26}, we had a somewhat different proof of Theorem \ref{intro:main} that was longer and instead relied on the Gaussian isoperimetric inequality and a second variation argument.

\section{First Variation Arguments}
\label{fd:foundations}

\begin{lemma}
\label{fd:attainment}
Fix $n\geq1$, $N\geq2$. There exists a measurable partition $\A\subset\R^n$ into $k\leq N$ sets maximizing $F(\A)$ among partitions of $\R^n$ into $N$ sets.  Let $k\leq N$ and let $\A=(A_1,\ldots,A_k)$ be any such maximizing partition, such that $\gamma_n(A_i)>0$ for all $1\leq i\leq k$.
Then
\begin{equation}\label{fd:obtuse-identities}
 \sum_{i=1}^k z_i=0,\qquad \langle z_i,z_j\rangle<0\quad\forall\,1\leq i<j\leq k.
\end{equation}
The span of $z_1,\ldots,z_k$ has dimension $k-1$, and, up to sets of measure zero,
\begin{equation}\label{fd:winning-cones}
 A_i
 =\{x\in\R^n\colon \langle z_i,x\rangle=\max_{j=1,\ldots,k}\langle z_j,x\rangle\}.
\end{equation}
After applying a rotation to $\R^n$ and projecting onto the span of $z_1,\ldots,z_k$, we may assume that $n=k-1$, and $k\geq2$.  That is, $\max_{\A\subset\R^n,|\A|\leq N}F(\A)=\max_{\A\subset\R^{k-1},|\A|=k}F(\A)$.
\end{lemma}
\begin{proof}  
Under the additional assumption that $\A\subset\R^n$ has the \textit{minimal} number of sets $k$ that maximizes $F(\A)$ in $\R^n$, this Lemma is proven in \cite[Lemmas 3.1--3.3]{KN}.  We therefore prove the Lemma without this assumption.  It suffices to show that $\langle z_i,z_j\rangle<0$ for all $1\leq i<j\leq k$, since the remaining steps of \cite[Lemmas 3.1--3.3]{KN} apply verbatim.  Also, since moving a set $E\subset A_i$ into $A_j$ increases $F$ by $2\int_{E}\langle z_j - z_i,x\rangle d\gamma_{n}(x)+2\|\int_{E}xd\gamma_n(x)\|^2$,  we have that $z_1,\ldots,z_k$ are distinct (otherwise choose $E$ with $\int_{E}xd\gamma_n(x)\neq0$ to increase $F$) and $A_i
 \subset\{x\in\R^n\colon \langle z_i,x\rangle=\max_{j=1,\ldots,k}\langle z_j,x\rangle\}$ $\forall$ $1\leq i\leq k$, so the $z_1,\ldots,z_k$ are all nonzero (since $F(\A)>0$, some $z_j$ is nonzero, hence at least two of them are nonzero since $\sum_{j=1}^{k}z_j=0$, so if $z_i=0$ then \eqref{fd:winning-cones} implies $\gamma_n(A_i)=0$, but this contradicts the assumption $\gamma_{n}(A_i)>0$).

Suppose $\A=(A_1,\ldots,A_k)\subset\R^{n}$ maximizes $F(\A)$ on $\R^{n}$.  It cannot occur that $\langle z_i,z_j\rangle>0$ for some $1\leq i<j\leq k$, since otherwise we could combine $A_i,A_j$ into one set $A_i\cup A_j$ whose moment $z_i+z_j$ would satisfy $\|z_i+z_j\|^2>\|z_i\|^2+\|z_j\|^2$, i.e. $\A$ would not maximize $F$.  We further claim: it cannot occur that $\langle z_i,z_j\rangle=0$ for some $1\leq i<j\leq k$.  Otherwise, consider $x(s)\colonequals z_i -sz_j$, $s\geq0$, and note that $\langle z_i,x(s)\rangle=\|z_i\|^2>0$, $\langle z_j,x(s)\rangle=-s\|z_j\|^2$, and for any other $1\leq\ell\leq k$, $\langle z_\ell,x(s)\rangle=\langle z_\ell,z_i\rangle-s\langle z_\ell,z_j\rangle$ is a nondecreasing function of $s$ that is nonpositive at $s=0$.  Also, $\sum_{\ell\neq i,j}\langle z_\ell, z_j\rangle=-\|z_j\|^2<0$ since $\sum_{q=1}^k z_q=0$, so there is at least one $\ell\neq i,j$ where $\langle z_\ell,x(s)\rangle$ is a strictly increasing function of $s$.  Let then $s_*>0$ be the smallest $s$ such that there is an $\ell\neq i,j$ where $\langle z_\ell,x(s)\rangle=\langle z_i,x(s)\rangle=\|z_i\|^2$.  By contrast, if we replace $A_i$ and $A_j$ by $A_i\cup A_j$, then $\langle z_i+z_j,x(s_*)\rangle=\|z_i\|^2 - s_*\|z_j\|^2<\|z_i\|^2$.  So, for $s$ close to $s_*$ with $s<s_*$, we have by \eqref{fd:winning-cones} that $x(s)\in A_i$, but if we replaced $A_i,A_j$ by $A_i\cup A_j$, then $F$ is increased if we merge a neighborhood of $x(s)$ with $A_\ell$.  In other words, $\A$ does not maximize $F(\A)$, a contradiction.  We conclude that \eqref{fd:obtuse-identities} holds.
\end{proof}
\begin{remark}\label{12rk}
\eqref{fd:winning-cones} implies that $p_{i}\leq1/2$ for all $1\leq i\leq k$.
\end{remark}

\begin{lemma}[{\cite[Lemma 3.7]{KN}}]\label{knlem}
Let $G\sim\gamma_{k-1}$.  For any $v_1,\ldots,v_k\in\R^{k-1}$, 
\[
(\mathbb E\max_{i=1,\ldots,k}\langle v_i,G\rangle)^2\le C_k\sum_{i=1}^k\|v_i\|^2,
\]
with equality when $v_i=z_i$ for all $1\leq i\leq k$ for some $\A\subset\R^{k-1}$ with $F(\A)=C_k$.
\end{lemma}

\begin{proof}[Proof of Remark \ref{sdrk}]
Note that $\sum_{i,j=1}^k\langle\int_{A_i}xd\gamma_{k-1}(x),\,\int_{A_j}xd\gamma_{k-1}(x)\rangle=\|\int_{\R^{k-1}}xd\gamma_{k-1}(x)\|^2=0$, so if $B$ has diagonal entries $1$ and off-diagonal entries $-1/(k-1)$,
\[
C(B)=C\Big(\frac{k}{k-1}I_{k}-\frac{11^{T}}{k-1}\Big)=\frac{k}{k-1}C(I_k)=\frac{1}{1-1/k}\frac{9}{8\pi},
\]
And if $\widetilde{B}=(\langle v_i, v_j\rangle)$ where $v_1,\ldots,v_k\in S^{k-1}$ and $\sum_{i=1}^{k}v_i=0$, then $\frac{1}{k(k-1)}\sum_{i\neq j}\langle v_i,v_j\rangle=-\frac{k}{k(k-1)}=-\frac{1}{k-1}$ so $\frac{1}{k!}\sum_{P\in S_k}P^T\widetilde{B}P=B$, so by convexity and permutation invariance, $C(B)\leq C(\widetilde{B})$, with equality only when $\widetilde{B}=B$.  To see this, assume $C(B)=C(\widetilde{B})$.  

For any $3$-tuple $i,j,\ell$ of distinct indices in $\{1,\ldots,k\}$, the partition $A_1,\ldots,A_k$ of $\R^{k-1}$ consisting of three $120$ degree sectors $\times \R^{k-3}$ with all other cells empty satisfies $c_{ij\ell}\colonequals\sum_{i',j'=1}^{k}\widetilde{b}_{i'j'}\langle\int_{A_{i'}}xd\gamma_{k-1}(x),\,\int_{A_{j'}}xd\gamma_{k-1}(x)\rangle\leq C(\widetilde{B})$ and $c_{ij\ell}=\frac{3}{8\pi}(3-\widetilde{b}_{ij}-\widetilde{b}_{i\ell}-\widetilde{b}_{j\ell})$, using $\int_{\theta=-\pi/3}^{\pi/3} \int_{r=0}^{\infty}r\cos\theta e^{-r^2 /2}rdrd\theta/(2\pi)=\sqrt{3/(8\pi)}$.  Since $\sum_{i<j}\widetilde{b}_{ij}=-k/2$, averaging gives $(1/\binom{k}{3})\sum_{i<j<\ell}c_{ij\ell}=\frac{3}{8\pi}(3 + 3/(k-1))=C(B)=C(\widetilde{B})$, implying that $c_{ij\ell}=C(B)$ for all $1\leq i<j<\ell\leq k$.  As in \eqref{fd:winning-cones}, the resulting optimality conditions \cite[Lemma 2.3]{khot13} imply that $\langle u_i - u_j,x\rangle=0$ for every $x$ in the interface between $A_i,A_j$, where $u_a\colonequals\sum_{s\in\{i,j,\ell\}}\widetilde{b}_{as}z_s$.  In particular, $\langle u_i - u_j, z_\ell\rangle=0$.  Since $u_i-u_j=(1-\widetilde{b}_{ij})(z_i - z_j)+(\widetilde{b}_{i\ell}-\widetilde{b}_{j\ell})z_\ell$, and $\langle z_i - z_j, z_\ell\rangle=0$, we get $\widetilde{b}_{i\ell}=\widetilde{b}_{j\ell}$ for all distinct $i,j,\ell$, which finally implies that $\widetilde{B}=B$.  So, $\frac{1}{C(\widetilde{B})}\leq\frac{1}{C(B)}=(1-1/k)\frac{8\pi}{9}$ with equality corresponding to the best constant in \eqref{svdgro}.
%
%
\end{proof}
%
%
%

\section{Three dimensions: Two Critical Partitions}
\label{sec:analytic-three}

We now show that, when $n=3$, the critical point equations \eqref{fd:winning-cones} have exactly two nondegenerate solutions (up to rotations of the sets, measure zero changes, and relabeling).

\subsection{Classification of Critical Points}

\begin{lemma}\label{an:classification}
Let $A_1,\ldots,A_4\subset\R^3$ satisfy the critical point condition \eqref{fd:winning-cones}.  Assume also that $\mathrm{span}(z_1,z_2,z_3,z_4)$ has dimension $3$, and $\gamma_3(A_i)>0$ $\forall$ $1\leq i\leq 4$.  Then the partition $A_1\cap S^2,\ldots,A_4\cap S^2$ consists of a regular spherical triangle with edge length $\theta$ together with three congruent isosceles triangles.  Moreover $\theta$ can take only two possible values.
\end{lemma}
\begin{proof}
Let $v_1,\dots,v_4\in S^2$ be the vertices of the partition ($\{v_i\}=S^2\cap\bigcap_{j\neq i}A_{i}$).  For each $1\leq i,j\leq 4$, denote $c_{ij}\colonequals\langle v_i,v_j\rangle$.  For any $c\in(-1,1)$, denote $f(c)\colonequals\arccos(c)/\sqrt{1-c^2}$.  By \cite[Lemma 3.14, Corollary 3.15]{HJN}
\begin{equation}\label{expeq}
0
=\sum_{j\in\{1,2,3,4\}\setminus\{i\}}f(c_{ij})(v_j-c_{ij}v_i),\qquad\forall\,1\leq i\leq 4.
\end{equation}
%
%


%

Let $V\in\R^{3\times 4}$ have columns $(v_i)_{i=1}^4$, and let $B\in\R^{4\times 4}$ be symmetric with
$B_{ij}\colonequals f(c_{ij})$ for $i\ne j$ and
$B_{ii}\colonequals-\sum_{j\ne i}f(c_{ij})c_{ij}$. Then $VB=0$ using $c_{ij}=c_{ji}$.

Since $v_1,\ldots,v_4$ span $\R^3$ by assumption (otherwise the span of $z_1,\ldots,z_4$ would not have dimension $3$), we have $\operatorname{rank}(V)=3$, so $\dim\ker(B)\ge 3$. Hence $\operatorname{rank}(B)\le 1$. Since $B\neq 0$, in fact $\operatorname{rank}(B)=1$. Therefore $\exists$ $x\in\R^4$ with all nonzero entries such that $B=\pm xx^{T}$.  In fact, $B=xx^T$, since $B_{12}B_{13}B_{23}=\pm x_{1}^2x_{2}^2x_{3}^2>0$, so $B\neq - xx^T$. (Since $B_{ij}>0$ $\forall$ $i\neq j$, we may assume $x_{i}>0$ for all $1\leq i\leq 4$.)  Let
\begin{equation}\label{psidef}
 \psi(\theta)\colonequals\frac{\theta}{\sin\theta}\qquad\forall\,0<\theta<\pi,
\end{equation}
so $\psi\colon(0,\pi)\to(1,\infty)$ is a real-analytic bijection (since $\psi'(\theta)=(\sin\theta - \theta\cos\theta)/\sin^{2}\theta$ so $\psi'(\theta)>0$ $\forall$ $0<\theta<\pi/2$ since $\tan\theta>\theta$ there, by the power series expansion of $\tan\theta$, so $\psi'>0$ extends to all $0<\theta<\pi$ since $(d/d\theta)(\sin\theta-\theta\cos\theta)=\theta\sin\theta>0$ $\forall$ $0<\theta<\pi$). Let
\begin{equation}\label{rhodef}
\rho(s)\colonequals\cos\bigl(\psi^{-1}(s)\bigr),
 \qquad
 \eta(s)\colonequals s\rho(s),\qquad\forall\,s>1.
\end{equation}

Write $s=\psi(\theta)$ with $\theta\in(0,\pi)$. Then by \eqref{rhodef} and \eqref{psidef},
\begin{equation}\label{etaeq}
 \eta(s)=\theta\cot\theta,\qquad\eta(s)+s=\theta\cot(\theta/2)
\end{equation}
%

Then $(d/d\theta)(\eta(s)+s)=(\sin(\theta)-\theta)/(2\sin^2(\theta/2))<0$.  Since $\psi'>0$, the chain rule implies $(d/ds)(\eta(s)+s)<0$, so that $\eta'(s)<-1$.  Let $\theta=\cos^{-1}(c)$.  Then $f(c)=\theta/\sin\theta$, and $\eta(f(c))=f(c)\rho(f(c))=f(c)\cos\theta=cf(c)$.  The diagonal entries of $B=xx^T$ say $B_{ii}=x_i^2$ and $f(c_{ij})=B_{ij}=x_ix_j$ for all $1\leq i<j\leq 4$, and the definition of $B_{ii}$ gives
\begin{equation}\label{xeqs}
x_{i}^{2}=-\sum_{j\neq i}f(c_{ij})c_{ij}
=-\sum_{j\neq i}\eta(f(c_{ij}))
=-\sum_{j\neq i}\eta(x_ix_j),\qquad\forall\,1\leq i\leq 4.
\end{equation}
%
%
%
%

Subtracting these equations for two different values of $i$ gives
\[
x_i^2 - x_j^2 = -\sum_{k\neq i,j}( \eta(x_i x_k) - \eta(x_jx_k)).
\]
The mean value theorem says $\eta(x_i x_k) - \eta(x_jx_k)=(x_i-x_j)x_k \eta'(\xi)$ for some $\xi$.  Since $\eta'<-1$,
\begin{equation}\label{aineq}
x_i+x_j > \sum_{k\neq i,j}x_k,\qquad\text{if}\,x_{i}\neq x_j.
\end{equation}

After reordering, assume $x_{1}\leq x_2\leq x_3\leq x_4$.  Then \eqref{aineq} rules out $x_1<x_3$ since $x_1+x_3\leq x_2+x_4$ but \eqref{aineq} would imply $x_1+x_3>x_2+x_4$.  Therefore,
\begin{equation}\label{xeq}
x_1=x_2=x_3\leq x_4.
\end{equation}

Let  $\theta_{ij}\colonequals\cos^{-1}(\langle v_i,v_j\rangle)$ for all $1\leq i<j\leq4$ by definition of $f$ and $B$.  Then $\theta_{ij}/\sin\theta_{ij}=f(c_{ij})=x_ix_j$.  Since $\psi$ is increasing, we get $\theta_{12}=\theta_{23}=\theta_{13}$ and $\theta_{14}=\theta_{24}=\theta_{34}$ by \eqref{xeq}.  That is, the partition $A_1\cap S^2,\ldots,A_4\cap S^2$ consists of a regular spherical triangle with edge length $\theta_{12}$ together with three congruent isosceles triangles with edge lengths $\theta_{12},\theta_{14},\theta_{14}$.

It remains to show that $\theta_{12}$ can take at most two values.

We could simply invoke \cite[page 9]{HJN}, but we instead include the argument here.  Let
\begin{equation}\label{abdef}
 \alpha\colonequals\rho(x_1^2),
 \qquad
 \beta\colonequals\rho(x_1x_4).
\end{equation}
Using \eqref{rhodef}, \eqref{xeqs} becomes
%
\begin{equation}
 x_1^2+2x_1^2\alpha+x_1x_4\beta=0,
 \qquad
 x_4^2+3x_1x_4\beta=0.
\end{equation}
Therefore
\begin{equation}\label{eq:alpha-beta}
1+3\frac{x_1}{x_4}\beta=0
,\quad
1+2\alpha+\frac{x_4}{x_1}\beta=0
=1+2\alpha - 3\beta^2,
\qquad\alpha=(3\beta^2 -1)/ 2.
\end{equation}
%
By \eqref{rhodef} and \eqref{abdef}, $x_1^2=\psi(\cos^{-1}\alpha)$ and $x_1x_4=\psi(\cos^{-1}\beta)$.  So, the relation $x_1x_4-(x_4/x_1)x_1^2=0$ gives the single equation
\begin{equation}\label{eq:Fbeta}
F(\beta)\colonequals\psi(\cos^{-1}\beta)+3\beta\,\psi\!\left(\cos^{-1}\frac{3\beta^2-1}{2}\right)=0,
 \qquad \beta\in(-1,0).
\end{equation}
%
%

The function $\delta\mapsto \sqrt{1-\delta^2}F(-\delta)$ is strictly convex on $\delta\in(0,1)$, since its second derivative is $
\frac{\delta(37-18\delta^2+45\delta^4)}
     {(1-\delta^2)^{3/2}(1+3\delta^2)^2}
+\frac{18\sqrt3\,\delta \cos^{-1}((3\delta^2 - 1)/2)}{(1+3\delta^2)^{5/2}}>0$, so it has at most two zeros.  In fact, exactly two zeros can be found numerically.  One is $1/3$, the other is larger than $1/2$.  ($\delta=1/3$ corresponds to $\theta_{12}=\cos^{-1}(-1/3)\approx 1.9106332362490187$, and the other corresponds to $\theta_{12}\approx 1.5379684120790425$ \cite{HJN}.)

\end{proof}

\section{Proof of the Main Theorem}\label{secmain}

\begin{proof}[Proof of Theorem \ref{intro:main} and Remark \ref{uniquerk}]
Fix $n,N\geq1$.  Let $A_1,\ldots,A_k$ be a partition of $\R^n$ with $k\leq N$ that maximizes $F(\A)$ on $\R^n$ with $\gamma_{n}(A_i)>0$ for all $1\leq i\leq k$.  From Lemma \ref{fd:attainment}, we may assume $k=n+1$ and $F(A_1,\ldots,A_k)=C_k$.  We use the notation of \eqref{mdef}.  If $k=2$, \cite[Corollary 3.4]{KN} implies $\sum_{i=1}^{2}\|z_i\|^2\leq\frac{1}{\pi}<9/(8\pi)$.  If $k=3$, \cite[Corollary 3.4]{KN} implies via a Lagrange multiplier argument that $\sum_{i=1}^{3}\|z_i\|^2\leq9/(8\pi)$ with equality only for three disjoint 120 degree sectors centered at the origin.  If $k=4$, Lemma \ref{an:classification} implies that there are only two different partitions that could maximize $F(\A)$ (up to symmetries), and an explicit calculation \cite[Lemma 4.1]{HJN} shows they satisfy $\sum_{i=1}^{4}\|z_i\|^2<9/(8\pi)$.  It remains to show: it cannot occur that there is a partition that maximizes $F(\A)$ with more than four sets of positive measure.  This is shown in Lemma \ref{dimlem}, concluding the proof.

\end{proof}

\section{Measure and Moment Constraints}\label{sc:section}

For any $t\in\R$ and $1\leq i\leq k$, define
\begin{equation}\label{rdef}
\phi(t)\colonequals e^{-t^{2}/2}/\sqrt{2\pi},\quad
 \Phi(t)\colonequals\int_{-\infty}^t \phi(u)du,\quad
 r_i\colonequals\frac{\|z_i\|},{\sqrt{C_k}}\qquad
 a=a_k\colonequals
\begin{cases}0.929,&k=5,\\0.884,&k\ge6.\end{cases}
\end{equation}
\begin{lemma}[Measure and Moment Constraints]\label{conlem1}
Let $k\geq5$.  Assume $\A=(A_1,\ldots,A_k)\subset\R^{k-1}$ achieves the supremum in the definition of $C_k$ with $\gamma_{k-1}(A_i)>0$ $\forall$ $1\leq i\leq k$.   Then, $\forall$ $1\leq i\leq k$, we have
\begin{equation}\label{sc:three-bounds}
 r_i\le\frac23,\qquad p_i\ge ar_i^2.
\end{equation}
\end{lemma}

\begin{proof}

The argument proceeds by rearrangement.
%
Let \(e_i\colonequals z_i/\|z_i\|\), \(t_i\colonequals\Phi^{-1}(1-p_i)\).  Let $H_i$ be the half space
\(H_i\colonequals\{y\in\R^{k-1}\colon\langle e_i,y\rangle\ge t_i\}\).  Since $\Phi^{-1}(1-p_i)=-\Phi^{-1}(p_i)$, we have $\gamma_{k-1}(H_i)=\Phi(-t_i)=p_i$.  That is, $H_i$ is a half space with the same measure as $A_i$.  Rearrangement in the direction $e_i$ then implies that
\[
\|z_i\|
=\int_{A_i}\langle x,e_i\rangle d\gamma_{k-1}(x)
\leq\int_{H_i}\langle x,e_i\rangle d\gamma_{k-1}(x)
=\int_{t_i}^{\infty} ud\gamma_1(u)
=\phi(t_i).
\]
\begin{equation}\label{sc:halfspace-ratio}
r_i\leq 2/3
,\qquad
 \frac{p_i}{r_i^2}
 \stackrel{\eqref{rdef}}{=}\frac{p_i C_k}{\|z_i\|^2}
 \geq\frac{p_i C_k}{(\phi(t_i))^2}
\stackrel{\eqref{mdef}\wedge\eqref{cst}}{\geq}
 \frac94 e^{t_i^2}\Phi(-t_i)
\end{equation}
The $2/3$ bound follows by $r_i=\|z_i\|/\sqrt{C_k}\leq\phi(t_i)/\sqrt{C_k}\leq\phi(0)/\sqrt{C_k}\leq 2/3$ by \eqref{cst}.

We now minimize the right side of \eqref{sc:halfspace-ratio}.  Put \(M(t)\colonequals\phi(t)/\Phi(-t)\) for any $t\in\R$.
Define \(f(t)\colonequals(9/4)e^{t^2}\Phi(-t)\).  Then $(\log f(t))'=2t-M(t)$ and $(\log f(t))''=2-M'(t)$, while
\[
M'(t) = \frac{\Phi(-t)\cdot(-t)\phi(t) + [\phi(t)]^2}{[\Phi(-t)]^2}
=[M(t)]^2 - tM(t).
\]
\begin{equation}\label{meq}
0
\leq\frac{1}{\Phi(-t)}\int_{t}^{\infty}(u-M(t))^2\phi(u)du
=tM(t)+1-[M(t)]^2
=1-M'(t).
\end{equation}
%
%
Consequently, $(\log f(t))''\geq1>0$, i.e. \((\log f)'\) is strictly
increasing. Its signs at \(0.6119\) and \(0.6121\) are opposite, so
\(f\) has a unique global minimum in that interval.
The elementary numerical bounds in Section \ref{sec:numA} below give
\begin{equation}\label{sc:halfspace-minimum}
 \min_{t\in\R}f(t)>0.8843837>0.884.
\end{equation}
This proves the last inequality of \eqref{sc:three-bounds} when \(k\ge6\).

To prove the remaining case $k=5$ of \eqref{sc:three-bounds}, we additionally use that $A_i$ is a cone by \eqref{fd:winning-cones}.  By Lemma \ref{fd:attainment}, $A_i\subset\R^4$.
Let $G\in\R^4$ be a standard Gaussian.  Then $G=RU$ where \(U\) is uniform on \(S^3\) and independent of \(R\).
The coordinate \(\langle e_i,U\rangle\) has density
\((2/\pi)\sqrt{1-u^2}\) on \([-1,1]\), while $\E R=3\sqrt{\pi/8}$.  Let $\sigma$ denote the Haar probability measure on $S^3$.

The spherical cap $S_\alpha\colonequals\{\theta\in S^{3}\colon \langle e_i,\theta\rangle\geq \alpha\}$ satisfies $p_i=\sigma(A_i\cap S^{3})=\sigma(S_\alpha)= (1/2)-(1/\pi)(\arcsin \alpha + \alpha\sqrt{1-\alpha^2})$.  Since $0<p_i\leq1/2$ by Remark \ref{12rk}, $0\leq\alpha<1$.  Since $A_i$ is a cone, spherical rearrangement gives
\begin{flalign*}
 \|z_i\|
 &=3\sqrt{\pi/8}\int_{A_i\cap S^{3}}\left\langle e_i,\theta\right\rangle d\sigma(\theta)
\leq3\sqrt{\pi/8}\int_{S_\alpha}\left\langle e_i,\theta\right\rangle d\sigma(\theta)\\
 &=3\sqrt{\frac{\pi}{8}}\frac{2}{\pi}\int_\alpha^1u(1-u^2)^{1/2}\,du
 =\frac{(1-\alpha^2)^{3/2}}{\sqrt{2\pi}}.
\end{flalign*}

\begin{equation}\label{pk5bd}
 \frac{p_i}{r_i^2}
 \stackrel{\eqref{rdef}}{=}\frac{p_i C_k}{\|z_i\|^2}
 \geq\frac{p_i C_k2\pi}{(1-\alpha^2)^3}
\stackrel{\eqref{mdef}\wedge\eqref{cst}}{\geq}\frac94 \frac{p_i}{(1-\alpha^2)^3}
 =\frac{9}{4}\frac{(1/2)-(1/\pi)(\arcsin \alpha + \alpha\sqrt{1-\alpha^2})}{(1-\alpha^2)^3}.
\end{equation}

As in the previous case, we now minimize the right side over $\alpha$.  Elementary numerical bounds in Section \ref{sec:numB} prove a lower bound of $.93$, completing the remaining $k=5$ case of the last inequality of \eqref{sc:three-bounds}.
\end{proof}

\begin{lemma}[More Constraints]\label{conlem2}
Let $k\geq5$.  Assume $\A=(A_1,\ldots,A_k)\subset\R^{k-1}$ achieves the supremum in the definition of $C_k$ with $\gamma_{k-1}(A_i)>0$ $\forall$ $1\leq i\leq k$.  Then
\begin{equation}\label{con2}
p_i\ge0.415r_i+0.15r_i^2,\qquad\forall\,1\leq i\leq k
\end{equation}
\end{lemma}
\begin{proof}
\textbf{Step 1. Removing One Set, Ehrhard's Inequality.}  
Fix $1\leq i\leq k$, let $G$ be a standard Gaussian in $\R^{k-1}$ and put \(X_j\colonequals\langle z_j,G\rangle\) for all $1\leq j\leq k$, and set \(\alpha\colonequals k/(k-1)\).  For any $j\neq i$, define $v_j\colonequals z_j + z_i/(k-1)$.  Then
\[\sum_{j\neq i}\|v_j\|^2=\frac{\|z_i\|^2}{k-1}
+\sum_{j\neq i}(\|z_j\|^2+2\langle z_j, z_i\rangle/(k-1))\stackrel{\eqref{fd:obtuse-identities}}{=}-\frac{\|z_i\|^2}{k-1}+\sum_{j\neq i}\|z_j\|^2.
\]

  Then Lemma \ref{knlem} says that
\begin{flalign*}
(\E \max_{j\colon j\neq i}X_j)^2
&=(\E \max_{j\colon j\neq i}\langle v_j, G\rangle)^2
\leq C_k\sum_{j\neq i}\|v_j\|^2
=C_k\Big(-\frac{\|z_i\|^2}{k-1}+\sum_{j\neq i}\|z_j\|^2\Big)\\
&\stackrel{\eqref{mdef}}{=}C_k(C_k - \alpha \|z_i\|^2)
\stackrel{\eqref{rdef}}{=}C_k^2(1-\alpha r_i^2).
\end{flalign*}

By assumption 
$C_k
=\E \max_{j=1,\ldots,k}X_j
=\E\max_{j\neq i}X_j+\E(X_i - \max_{j\neq i}X_j)1_{\{X_{i}>\max_{j\neq i}X_j\}}$.  Rearranging this identity yields
\begin{equation}\label{sc:loss-lower}
 \mathbb E\max(X_i-\max_{j\ne i}X_j,\,0)
 \ge C_k\bigl(1-\sqrt{1-\alpha r_i^2}\bigr).
\end{equation}

The vectors $\{z_i-z_j\}_{j\neq i}$ form a basis of $\R^{k-1}$, since $\{z_j\}$ span $\R^{k-1}$ while \eqref{fd:obtuse-identities} implies $z_i = -(1/k)\sum_{j\neq i}(z_j - z_i)$ by Lemma \ref{fd:attainment}, so any $x\in\R^{k-1}$ can be written as a sum of $z_j-z_i$ terms plus a $z_i$ term.  Let then $h\in\R^{k-1}$ satisfy
\(\langle h,z_i-z_j\rangle=1\) for every \(j\ne i\).
Define
%
%
\[
 p_i(u)\colonequals
 \P(X_i\ge X_j+u,\,\,\forall\,j\ne i)
\stackrel{\eqref{fd:winning-cones}}{=}\gamma_{k-1}(A_i+uh),\qquad\forall\,u\geq0.
\]
Note that $k-1=\sum_{j\neq i}\langle h,z_i - z_j\rangle=k\langle h,z_i\rangle$ by \eqref{fd:obtuse-identities}, so \(\langle h,z_i\rangle=(k-1)/k=1/\alpha\).
Differentiating $p_i(u)=\int_{A_i + uh}d\gamma_{k-1}(x)=\int_{A_i}d\gamma_{k-1}(x+uh)
=\int_{A_i}e^{-u\langle h,x\rangle-\|uh\|^2/2}d\gamma_{k-1}(x)$,
\begin{equation}\label{sc:translation-derivative}
 p_i'(0)=-\int_{A_i}\langle h,x\rangle\,d\gamma_{k-1}(x)
\stackrel{\eqref{mdef}}{=}-\langle h,z_i\rangle=-1/\alpha.
\end{equation}

Ehrhard's inequality \cite{Ehrhard} states that \(\Phi^{-1}\circ\gamma_{k-1}\) is
concave under Minkowski interpolation $\lambda A + (1-\lambda)B$ of convex sets $A,B\subset\R^{k-1}$, $0<\lambda<1$.  
The cone \(A_i\) is closed and convex. Since $\{z_i - z_j\}_{j\neq i}$ is a basis, $A_i$ has nonempty interior and is a proper subset of $\R^{k-1}$ with positive Gaussian measure. Ehrhard's inequality implies that
\begin{flalign*}
\Phi^{-1}(p_i(\lambda u + (1-\lambda)v))
&=\Phi^{-1}(\gamma_{k-1}(\lambda(A_i + uh)+(1-\lambda)(A_i +vh)))\\
&\geq \lambda\Phi^{-1}(p_i(u))+(1-\lambda)\Phi^{-1}(p_i(v))\qquad\forall\,u,v\geq0,\,
0<\lambda<1.
\end{flalign*}
That is, 
\(u\mapsto\Phi^{-1}(p_i(u))\) is then concave on $u\geq0$.  Note that $0<\gamma_{k-1}(A_i +uh)<1$, and $p_i'(0)$ exists by \eqref{sc:translation-derivative}. Let
\(q\colonequals\Phi^{-1}(p_i)\).  Concavity implies that $u\mapsto\Phi^{-1}(p_i(u))$ lies below its tangent line at $u=0$.  Since $(d/du)\Phi^{-1}(p_i(u))|_{u=0}=p_i'(0)/\Phi'(\Phi^{-1}p_i(0))=-1/(\alpha\phi(q))$, 
\begin{equation}\label{pineq}
 p_i(u)\le\Phi\!\left(q-\frac{u}{\alpha\phi(q)}\right),
 \qquad \forall\,u\geq0.
\end{equation}
Integration, using \((r\Phi(r)+\phi(r))'=\Phi(r)\) and the chain rule, gives
\begin{equation}\label{sc:loss-upper}
 \mathbb E\max(X_i-\max_{j\ne i}X_j,0)
 =\int_0^\infty p_i(u)\,du
\stackrel{\eqref{pineq}}{\le}\alpha\phi(q)\bigl(q\Phi(q)+\phi(q)\bigr).
\end{equation}
%
Define
\[
f(q)
\colonequals\frac{\phi(q)\bigl(q\Phi(q)+\phi(q)\bigr)}{\Phi(q)^2},\qquad\forall\,q\in\R.
\]
Combining \eqref{sc:loss-lower} and \eqref{sc:loss-upper}, rationalizing,
and canceling \(\alpha\), we obtain (using $\alpha>1$)
\begin{equation}\label{sc:loss-constraint}
 p_i^2f\bigl(\Phi^{-1}(p_i)\bigr)
 \ge\frac{C_kr_i^2}{1+\sqrt{1-\alpha r_i^2}}
 \ge\frac{(9/(8\pi))r_i^2}{1+\sqrt{1-r_i^2}}.
\end{equation}
%

\textbf{Step 2. Monotonicity.}  We show \(f\) is nonincreasing, satisfies \(0<f(q)<1\), and
\[
 p\mapsto p^2f\bigl(\Phi^{-1}(p)\bigr)
\]
is increasing on \((0,1/2]\).

Fix $q\in\R$.  Put \(\lambda\colonequals\phi(q)/\Phi(q)\). Integration of a standard Gaussian $Z$ gives
\[
 \mathbb E(Z\mid Z\le q)=-\lambda,\qquad
 v(q)\colonequals\operatorname{Var}(Z\mid Z\le q)=1-q\lambda-\lambda^2.
\]
The last identity was essentially shown in \eqref{meq}.
 Thus \(f=1-v=\lambda(q+\lambda)\).  The density of $Z$ is strictly positive, so that $v>0$, i.e. $f<1$.  The random variable \(U\colonequals q-Z\) satisfies $0<\E(U\mid Z\leq q)=q+\lambda$, implying that $f>0$.  Thus $0<f<1$.
 
Define
\(s(u)\colonequals\Phi(q-u)/\Phi(q)\), \(u\ge0\).
The identity
\((\log\Phi)''(q)=-\lambda(q+\lambda)<0\)
shows that \(\log s\) is concave. Its increments therefore satisfy
\[
 \log s(b+c)-\log s(c)\le\log s(b)-\log s(0),\qquad\forall\,b,c>0.
\]
Since \(s(0)=1\), we get \(s(b+c)\le s(b)s(c)\).
Since $s(u)=\P(U\geq u\mid Z\leq q)$, substituting $u=b+c$ gives
\begin{flalign*}
 \frac12\mathbb E(U^2\mid Z\leq q)
 &=\int_{0}^{\infty}us(u)du
 =\int_{0}^{\infty}(\int_{0}^{u}db) s(u)du
=\int_{0}^{\infty}\int_{b}^{\infty} s(u)dudb\\
&=\int_{0}^{\infty}\int_{0}^{\infty} s(b+c)dbdc
 \leq\int_0^\infty\!\int_0^\infty s(b)s(c)\,db\,dc
 =(\E(U\mid Z\leq q))^2.
\end{flalign*}
%
Hence \(v(q)\le(2-1)(\E(U\mid Z\leq q))^2=(q+\lambda)^2\). Then $v=(q+\lambda)'$ and $\lambda'=-\lambda(q+\lambda)$ imply
\[
 v'(q)=\lambda\bigl((q+\lambda)^2-v(q)\bigr)\ge0,
\]
so \(f=1-v\) is nonincreasing.
Finally, if $q\colonequals\Phi^{-1}(p)\leq0$, then $p\leq1/2$ and $p^2 f(\Phi^{-1}(p))=\phi(q)(q\Phi(q)+\phi(q))$ by definition of $f$.  When $q\leq0$, both $q\mapsto \phi(q)$ and $q\mapsto q\Phi(q)+\phi(q)$ have nonnegative derivative, completing the proof, since $q\Phi(q)+\phi(q)=\Phi(q)(q+\lambda)>0$.
%
%

\textbf{Step 3.  Conclusion.  Combining Estimates.}  Set
\(p_0(r)\colonequals(0.415+0.15r)r\) for all \(0<r\le2/3\). Then
\(p_0(r)\le p_0(2/3)=103/300<1/2\).
We will prove
\begin{equation}\label{sc:test-probability}
 (0.415+0.15r)^2f\bigl(\Phi^{-1}(p_0(r))\bigr)
 <\frac{9/(8\pi)}{1+\sqrt{1-r^2}},\qquad\forall\,0< r\leq 2/3.
\end{equation}
Assuming for now that \eqref{sc:test-probability} holds, multiply both sides of \eqref{sc:test-probability} by $r^2$ to get $h(p_0(r))\leq (9/(8\pi)) r^2/(1+\sqrt{1-r^2})$, where $h(p)\colonequals p^2 f(\Phi^{-1}(p))$, while \eqref{sc:loss-constraint} says $h(p_i)\geq (9/(8\pi))r_i^2 / (1+\sqrt{1-\alpha r_i^2})$.  Choosing $r=r_i$, we then get $h(p_0(r_i))\leq h(p_i)$.  From Step 2, $h$ is increasing, so $p_0(r_i)\leq p_i$, which concludes the proof (noting $p_i\leq1/2$ by Remark \ref{12rk} and $r_i\leq2/3$ by \eqref{sc:three-bounds}, so $p_0(r_i)\leq1/2$ as well).

We verify \eqref{sc:test-probability} by the Python code in Section \ref{secpyb} below, which splits the interval $[0,2/3]$ into sub-intervals $[a,b]\subset[0,2/3]$.  Then \eqref{sc:test-probability} holds for all $r\in[a,b]\cap(0,2/3]$ when the following quantity is positive:
\[
 \Delta(a,b,U)\colonequals
 \frac{9/(8\pi)}{1+\sqrt{1-a^2}}-U(0.415+0.15b)^2,
 \qquad U\colonequals
 \begin{cases}
 f(\Phi^{-1}(p_0(a))) & \text{, if }a>0\\
 1 & \text{, if }a=0.
 \end{cases}
\]
The implication $\Delta(a,b,U)>0$ $\Rightarrow$ \eqref{sc:test-probability} holds on $[a,b]\cap(0,2/3]$ follows by monotonicity (e.g. that $f$ is decreasing).

\end{proof}

\section{Constraints on Two Cells}
\label{geo:section}

In this section we assume $\A=(A_1,\ldots,A_k)$ is a partition of $\R^{k-1}$ with $F(\A)=C_k$ and $\gamma_{k-1}(A_i)>0$ for all $1\leq i\leq k$.  We now use the
first-moment identity for the cell with smallest centroid norm to obtain
constraints on pairs of the other centroids.
Let $k\geq5$.  Relabel the nonzero centroids so that
\begin{equation}\label{cldef}
 r_1\ge r_2\ge\cdots\ge r_k=t>0,
 \qquad
 L\colonequals r_1^2+r_2^2,
 \qquad c\colonequals r_3^2.
\end{equation}

Let $G\in\R^{k-1}$ be a standard Gaussian vector, and let
\[
 e\colonequals\frac{z_k}{|z_k|},\qquad T\colonequals\langle e,G\rangle,
 \qquad
 \alpha_j\colonequals-\frac{\langle z_k,z_j\rangle}{|z_k|^2}
 \quad\forall\,1\le j<k.
\]
Then \eqref{fd:obtuse-identities} implies
\begin{equation}
 \alpha_j>0,\quad\forall\,1\leq j<k,\qquad \sum_{j<k}\alpha_j=1.
 \label{geo:coefficients}
\end{equation}
Define the residual vectors and scores by
\[
 w_j=\frac{z_j}{\sqrt{C_k}}+\alpha_jte,
 \qquad Y_j=\langle w_j,G\rangle,\qquad\forall\,1\le j<k\]
Each $w_j$ is orthogonal to $e$.  Consequently $Y\colonequals(Y_j)_{j<k}$ is a
centered Gaussian vector independent of the standard normal variable
$T$, and
\begin{equation}
 \frac{\langle z_j,G\rangle}{\sqrt{ C_k}}=-\alpha_jtT+Y_j.
 \label{geo:decomposition}
\end{equation}
%
%
The random vector $Y$ has covariance matrix $H$ where
\begin{equation}
 H_{j\ell}=\frac{\langle z_j,z_{\ell}\rangle}{C_k}-\alpha_j\alpha_{\ell}t^2,
 \qquad h_j\colonequals H_{jj}=r_j^2-\alpha_j^2t^2.
 \label{geo:covariance}
\end{equation}
Since $\sum_{j<k}w_j=0$ by definition of $w_j$ and \eqref{fd:obtuse-identities}, every row of $H$ sums to zero; also $H_{j\ell}<0$
for $j\ne \ell$ by \eqref{fd:obtuse-identities}.  For distinct $j,\ell<k$, set
\begin{equation}\label{ddef}
 D_{j\ell}\colonequals\sqrt{h_jh_\ell-H_{j\ell}^2}.
\end{equation}
The quantity inside the root is nonnegative since it is the determinant of a $2\times 2$ minor of a covariance matrix.  Moreover, these quantities are positive.  Indeed, a linear dependence between
$w_j$ and $w_\ell$ would give a dependence between $z_j,z_\ell,z_k$, whereas
every proper subset of the centroids is linearly independent and
$k\ge5$.  In particular, all $h_j$ are positive.

\begin{lemma}[Two competing scores]
\label{geo:pair}
For every pair of distinct indices $j,\ell<k$,
\begin{equation}
 D_{j\ell}\le \frac3\pi\,t(1+\alpha_j)(1+\alpha_\ell).
 \label{geo:pair-bound}
\end{equation}
\end{lemma}

\begin{proof}
By definition of $z_k$, we have $\E G1_{A_k}=z_k$.  The definitions of $T$ and $Y_j$ yield
\begin{equation}
 \mathbb E(\mathbf 1_{A_k}T)=\|z_k\|=\sqrt{C_k}\,t\geq\frac{3}{2\sqrt{2\pi}}t,
 \quad
 \mathbb E(\mathbf 1_{A_k}Y_i)=\langle w_i,z_k\rangle=\langle w_i,e\rangle\|z_k\|=0,\,\forall\,i<k.
 \label{geo:centroid-identities}
\end{equation}
When $G\in A_k$, we have by \eqref{fd:winning-cones} that $T\geq0$.  And by \eqref{geo:decomposition}, when $G\in A_k$, we have $Y_j\le(1+\alpha_j)tT$ for all $j<k$.
Fix distinct $j,\ell<k$ and write
\[
 U\colonequals\frac{Y_j}{(1+\alpha_j)t},\qquad
 V\colonequals\frac{Y_\ell}{(1+\alpha_\ell)t},
 \qquad
 \Sigma\colonequals
 \begin{pmatrix}
 \frac{h_j}{t^2(1+\alpha_j)^2}
 &\frac{H_{j\ell}}{t^2(1+\alpha_j)(1+\alpha_\ell)}\\
  \frac{H_{j\ell}}{t^2(1+\alpha_j)(1+\alpha_\ell)}
 &\frac{h_\ell}{t^2(1+\alpha_\ell)^2}
 \end{pmatrix}.
\]
Then $(U,V)$ is independent of $T$, and its joint Gaussian density is
bounded above by
\begin{equation}
 M\colonequals\frac{1}{2\pi\sqrt{\mathrm{det}\Sigma}}
 =\frac{t^2(1+\alpha_j)(1+\alpha_\ell)}{2\pi D_{j\ell}}.
 \label{geo:density-bound}
\end{equation}
Using $\E 1_{A_k}U=\E 1_{A_k}V=0$ by \eqref{geo:centroid-identities}, $U,V\leq T$ on $A_k$, taking a
positive part, and then enlarging the integration region yields
\[
 \mathbb E[\mathbf 1_{A_k}T]
 =\mathbb E[\mathbf 1_{A_k}(T+U+V)]\\
 \le\mathbb E\big[
 \mathbf 1_{\{T\ge0,\ U\le T,\ V\le T\}}\max(T+U+V,0)\big].
\]
Fix $\tau\ge0$.  In the $u,v$ plane with $u\le \tau$, $v\le \tau$ the integrand
is supported on the triangle with vertices $(\tau,\tau)$, $(\tau,-2\tau)$, and
$(-2\tau,\tau)$. The corresponding unweighted integral is
\begin{align*}
 \iint_{\substack{u+v\geq-\tau,\\ u\leq\tau, v\leq\tau}}\max(\tau+u+v,0)\,dv\,du
 =\int_{z=0}^{z=3\tau}z(3\tau-z)dz
 =\frac{9\tau^3}{2},
\end{align*}
using $z=\tau+u+v$.
Using \eqref{geo:density-bound},
independence, and $\mathbb E[\max(T,0)^3]=2/\sqrt{2\pi}$, we obtain
\[
 \frac{3}{2\sqrt{2\pi}}t
 \stackrel{\eqref{geo:centroid-identities}}{\leq}\mathbb E[\mathbf 1_{A_k}T]
 \le \frac{9M}{2}\,\mathbb E[\max(T,0)^3]
 \stackrel{\eqref{geo:density-bound}}{=}\frac{9t^2(1+\alpha_j)(1+\alpha_\ell)}{2\pi\sqrt{2\pi}\,D_{j\ell}}.
\]
This implies
\eqref{geo:pair-bound}.
\end{proof}

Recall that $h_j>0$ and $H_{j\ell}<0$ for $j\ne \ell$.
Define $\rho_{j\ell}\colonequals H_{j\ell}/\sqrt{h_jh_\ell}$ for all $1\leq j,\ell\leq3$.  Then $\mathrm{det}(\rho_{j\ell})\geq0$ implies
$\sum_{1\leq j<\ell\leq3}\rho_{j\ell}^2\leq\sum_{1\leq j<\ell\leq3}\rho_{j\ell}^2+2|\rho_{12}\rho_{13}\rho_{23}|\le1$, so
\begin{equation}\label{djineq}
\begin{aligned}
 \sum_{1\leq j<\ell\leq3}D_{j\ell}^2
 &\stackrel{\eqref{ddef}}{=} \sum_{1\leq j<\ell\leq3}(h_{j}h_\ell - h_jh_\ell\rho_{j\ell}^2)
 \geq
 \sum_{1\leq j<\ell\leq3}h_jh_\ell-\max_{1\leq j<\ell\leq3}h_jh_\ell\\
 &
 =\min_{1\leq i\leq 3} \Big(h_i\sum_{j\in\{1,2,3\}\setminus\{i\}}h_j\Big)
 \stackrel{\eqref{cldef}}{\ge}(c-t^2)(L-2t^2).
\end{aligned}
\end{equation}
The last step used monotonicity, i.e. replacing $h_1,h_2,h_3$ by $r_1^2 - t^2,r_2^2 - t^2,r_3^2 - t^2$ decreases the penultimate quantity by \eqref{geo:coefficients} and \eqref{geo:covariance}; also $r_3^2-t^2=c-t^2$ and $h_1+h_2\geq L-2t^2$.

We now bound $\sum_{1\leq j<\ell\leq3}D_{j\ell}^2$ from above by \eqref{geo:pair-bound}.  
Monotonicity of $\sum_{1\leq j<\ell\leq3}(1+\alpha_j)^2(1+\alpha_\ell)^2$ implies we may assume $\sum_{j=1}^3 \alpha_j=1$.  Let $q\colonequals\sum_{1\leq j<\ell\leq 3}\alpha_j\alpha_\ell$, $u\colonequals \alpha_1\alpha_2\alpha_3\ge0$.  Using $\sum_{j=1}^{3}\alpha_j=1$, we get $\sum_{j=1}^3\alpha_j^2=1-2q$, $\sum_{1\leq j<\ell\leq 3}(\alpha_j^2\alpha_\ell + \alpha_\ell^2 \alpha_j)=q-3u$, $\sum_{1\leq j<\ell\leq 3}\alpha_j^2 \alpha_\ell^2 = q^2-2u$,
\begin{equation}\label{dj9}
 \frac{\pi^2}{9t^2}\sum_{1\leq j<\ell\leq3}D_{j\ell}^2
 \stackrel{\eqref{geo:pair-bound}}{\leq}
 \sum_{1\leq j<\ell\leq3}(1+\alpha_j)^2(1+\alpha_\ell)^2=9+2q+q^2-8u\le\frac{88}{9}.
\end{equation}
The last inequality used $q=(1-\sum_{j=1}^{3}\alpha_j^2)/2\leq1/3$ and $u\geq0$.
%
%
Since $88/\pi^2<9$ and $L\ge2c$,
\begin{equation}\label{det}
 \boxed{(c-t^2)(L-2t^2)\stackrel{\eqref{djineq}\wedge\eqref{dj9}}{<}9t^2}\,\,,\qquad
 \boxed{c<t^2+2.15t},\quad\text{since }3\sqrt{2}/2<2.15.
 \end{equation}

 \section{The Number of Sets is at Most 4}\label{seclow}

\begin{lemma}\label{dimlem}
Let $(A_1,\ldots,A_k)\subset\R^{k-1}$ be a partition of measurable sets each with positive measure such that $F(\A)=C_k$.  Then $k\leq4$. 
\end{lemma}
\begin{proof}
We argue by contradiction and assume $k\geq5$.  Using \eqref{sc:three-bounds} for $j=1,2$ gives $p_1+p_2\geq.884 L$ by \eqref{cldef}.  For each $j\geq3$, $r_j\leq r_3=\sqrt{c}$ by \eqref{cldef}, so $r_j\geq r_j^2/\sqrt{c}$, and Lemma \ref{conlem2} says $p_j\geq(.415/\sqrt{c}+.15)r_j^2$. Summing over $j$ and recalling that $1=\sum_{j=1}^{k}p_j=\sum_{j=1}^{k}r_j^2$ by \eqref{mdef} and \eqref{rdef}, so $\sum_{j\geq3}r_j^2=1-L$ by \eqref{cldef},
\[
1=\sum_{j=1}^{k}p_j
\geq .884L+(.415/\sqrt{c}+.15)(1-L)
=.884+(.415/\sqrt{c}+.15-.884)(1-L).
\]
Rearranging gives
\begin{equation}\label{mass}
 (1-L)(.415/\sqrt c-.734)\le.116.
\end{equation}
Since $L\le8/9$ by \eqref{sc:three-bounds}, \eqref{mass} implies \begin{equation}\label{cbd}
\sqrt c>.23.
\end{equation}
Moreover
\begin{equation}\label{tineq}
t>.06.
\end{equation}
Otherwise \eqref{det} gives $c<.1326<.37^2$ and \eqref{mass} gives $L>.7$,
contradicting \eqref{det} because
\[
 (.23^2-.06^2)(.7-2\cdot.06^2)=.03415504>9\cdot.06^2=.0324.
\]
Recalling $a$ from \eqref{rdef}, let $d\colonequals a-.15$, $s\colonequals .415/d$, $e(r)\colonequals .415r-dr^2$, and $E\colonequals 1-a$.
\begin{equation}\label{excess}
 \delta_j\colonequals 
 p_j-ar_j^2\stackrel{\eqref{sc:three-bounds}\wedge\eqref{con2}}{\ge}\max(0,e(r_j)),\qquad\sum_{j=1}^k\delta_j\stackrel{\eqref{rdef}}{=}E.
\end{equation}

We claim $r_4\le.4$. \begin{proof}[Proof of claim] Otherwise $c>.16$ by \eqref{cldef}, so \eqref{det} forces $t>.07$.  Write $e(r)=dr(s^2-r^2)/(s+r)$.  Let $K\colonequals.4d/(s+.4)$.  When $r\ge.4$, we claim
$\max(0,e(r))\ge K(s^2-r^2)$: for $r>s$ the right side is negative, and for $r<s$, it follows from our expression for $e(r)$.  So,
\begin{equation}\label{kineq}
\sum_{j=1}^{4}\delta_j
\geq K\sum_{j=1}^{4}(s^2 - r_j^2)
=K(4s^2-1+\sum_{j=5}^{k}r_j^2)
\geq K(4s^2-1+(k-4)(.07)^2).
\end{equation}
Since $t>.07$, for $j\ge5$, $.07<r_j\le1/\sqrt5<.45$ (since $\sum_{j=1}^{k}r_j^2=1$, by \eqref{rdef}, $r_{j}^2\leq 1/5$ for all $j\geq5$);
concavity of $e(u)$ and $e(.45)>e(.07)$ give $\delta_j\ge e(.07)$ by \eqref{excess}.
Therefore
\[
E\stackrel{\eqref{excess}}{=}
\sum_{j=1}^{4}\delta_j+\sum_{j=5}^{k}\delta_j
\stackrel{\eqref{kineq}}{\geq}
 K[4s^2-1+(k-4)(.07)^2]+(k-4)e(.07).
\]
For $k=5$ this exceeds
$.334(.1352+.0049)+.0252329=.0720263>.071=E$, by definition of $E$;
for $k\ge6$ it exceeds
$.304(.2784+2\cdot.0049)+2(.0254534)=.1385196>.116=E$, by definition of $E$.  Having achieved a contradiction $E>E$, we conclude $r_4\le.4$.

%
%
%
\end{proof}

%
Put $m_0=m_0(k)\colonequals 2$ for $k=5$ and $m_0\colonequals3$ otherwise. There are at least $m_0$ radii $r_j$ with $j>3$, all in $[t,.4]$. If $t\ge.12$, concavity of $e(r)$ would give
$E\ge m_0\min(e(.12),e(.4))>E$.
Thus $.06<t<.12$ by \eqref{tineq}. Since $e$ increases on $[0,.12]$ and $e(.4)>e(.12)$,
each of these $m_0$ cells contributes at least $e(t)$ to $\sum_{j=1}^{k}\delta_j$.

By \eqref{cbd}, \eqref{det} and $t<.12$, we get $.23<\sqrt c<\sqrt{2.3t}$, i.e. $.23<r_3<\sqrt{2.3t}$ by \eqref{cldef}, so
\[
 E\stackrel{\eqref{excess}}{\geq}
 \delta_3+
\sum_{j=4}^{3+m_0}\delta_j
\stackrel{\eqref{excess}}{\ge}\min(e(.23),e(\sqrt{2.3t}))+m_0e(t).
\]
The first alternative increases with $t$ and exceeds $E$ at $t=.06$.
The second,
$F(t)=.415\sqrt{2.3t}-2.3dt+m_0(.415t-dt^2)$, is concave,
so check only the endpoints of $[.06,.12]$:
\[
\begin{array}{c|ccc}
 &e(.23)+m_0e(.06)-E&F(.06)-E>&F(.12)-E>\\\hline
 k=5&.0274321&.0192392&.0069608\\
 k\ge6&.0073942&.0030308&.0149072
\end{array}
\]

%
%

%

That is, we showed $E>E$, a contradiction.  We conclude that $k\leq4$.

\end{proof}

\appendix

\section{Numerical Certificates}\label{sec:num}

\subsection{Numerical Bound 1}\label{sec:numA}
This Python code verifies \eqref{sc:halfspace-minimum}.

\begin{verbatim}
from math import exp, erfc, pi, sqrt
from scipy.optimize import minimize_scalar

f = lambda x: (9/8)*exp(x*x)*erfc(x/sqrt(2))
g = lambda x: 2*x - sqrt(2/pi)*exp(-x*x/2)/erfc(x/sqrt(2))

a, b = 0.6119, 0.6121
assert g(a) < 0 < g(b)

result = minimize_scalar(
    f, bounds=(a, b), method="bounded", options={"xatol": 1e-14}
)
assert result.success
assert result.fun > 0.8843837 > 0.884

print(f"x = {result.x:.6f}")
print(f"minimum = {result.fun:.12f}")
\end{verbatim}






\subsection{Numerical Bound 2}\label{sec:numB}

The following code proves the lower bound of $.93$ for the quantity \eqref{pk5bd}.

\begin{verbatim}
from math import asin, pi, sqrt
from scipy.optimize import minimize_scalar

def f(a):
    return (9/4) * (0.5 - (asin(a) + a*sqrt(1-a*a))/pi) / (1-a*a)**3

result = minimize_scalar(
    f, bounds=(0, 1), method="bounded", options={"xatol": 1e-14}
)
assert result.success
print(f"a = {result.x:.10f}")
print(f"minimum = {result.fun:.10f}")
\end{verbatim}

\subsection{Numerical Bound 3}\label{secpyb}

The following Python code verifies \eqref{sc:test-probability}.

\begin{verbatim}
import numpy as np
from scipy.stats import norm

r = np.linspace(0, 2/3, 1001)
a, b = r[:-1], r[1:]

P = (0.415 + 0.15*a[1:])*a[1:]
q = norm.ppf(P)
lam = norm.pdf(q)/P
U = np.r_[1.0, lam*(q + lam)] 

Delta = ((9/(8*np.pi))/(1 + np.sqrt(1-a*a)) - U*(0.415 + 0.15*b)**2)

assert np.all(Delta > 0)
print(f"Smallest margin: {Delta.min():.12f}")
\end{verbatim}

\bibliographystyle{abbrv}

\end{document}